\documentclass[12pt,reqno]{amsart}
\usepackage{amsfonts,amssymb,latexsym,amsmath,amsthm,color,dsfont}
\usepackage{enumerate,enumitem,cite}
\usepackage{stmaryrd}
\usepackage{fullpage}
\usepackage{multirow}
\usepackage{makecell,array}
\usepackage{tablefootnote}

\usepackage{tabularx}
\usepackage{graphicx, amsmath, amssymb, amsthm,anysize,float,booktabs,geometry}
\usepackage{pdflscape}
\usepackage{caption} 

\usepackage{hyperref}

\hypersetup{citecolor=red, linkcolor=blue, colorlinks=true}

\marginsize{2cm}{2cm}{2cm}{0cm}
\allowdisplaybreaks

\newcommand\F{\mathbb{F}}

\newcommand{\mc}{\mathcal}

\theoremstyle{plain}
\newtheorem{theorem}{Theorem}[section]
\newtheorem{problem}[theorem]{Problem}
\newtheorem{lemma}[theorem]{Lemma}

\newtheorem{proposition}[theorem]{Proposition}
\newtheorem{example}[theorem]{Example}

\newtheorem{remark}[theorem]{Remark}
\numberwithin{equation}{section}

\def\<{\langle}
\def\>{\rangle}

\title{Counterexamples, Constructions, and Nonexistence Results for Optimal Ternary Cyclic Codes}
\author{Jingjun Bao, Hanlin Zou$^*$}
\thanks{$^*$Corresponding author}

\address{Jingjun Bao, School of Mathematics and Statistics, Ningbo University, Ningbo 315211, China}
\email{baojingjun@hotmail.com}

\address{Hanlin Zou, School of Mathematics and Statistics, Yunnan University, Kunming 650091, China}
\email{zouhanlin@ynu.edu.cn}

\keywords{Cyclic code; Optimal code; Ternary code; Sphere packing bound.\\
\indent 2020 Mathematics Subject Classification: 94B15, 11T71}

\begin{document}
\maketitle

\begin{abstract}
Cyclic codes are an important subclass of linear codes with wide applications in communication systems and data storage systems. In 2013, Ding and Helleseth presented nine open problems on optimal ternary cyclic codes $\mc{C}_{(1,e)}$. While the first two and the sixth problems have been fully solved, others remain open. In this paper, we advance the study of the third and fourth open problems by providing the first counterexamples to both and constructing two families of optimal codes under certain conditions, thereby partially solving the third problem. Furthermore, we investigate the cyclic codes $\mc{C}_{(1,e)}$ where $e(3^h\pm 1)\equiv\frac{3^m-a}{2}\pmod{3^m-1}$ and $a$ is odd. For $a\equiv 3\pmod{4}$, we present two new families of optimal codes with parameters $[3^m-1,3^m-1-2m,4]$, generalizing known constructions. For $a\equiv 1\pmod{4}$, we obtain several nonexistence results on optimal codes $\mc{C}_{(1,e)}$ with the aforementioned parameters revealing the constraints of such codes. 
\end{abstract}

\section{Introduction}

Let $p$ be a prime and $m$ a positive integer. Let $\F_{p^m}$ be the finite field of order $p^m$. An $[n,k]$ linear code over $\F_p$ is a $k$-dimensional subspace of $\F_p^n$, and an $[n,k,d]$ linear code is an $[n,k]$ code in which the Hamming weight of each nonzero codeword is at least $d$. An $[n,k]$ code $\mc{C}$ is called {\it cyclic} if $(c_0,c_1,\ldots,c_{n-1})\in\mc{C}$ implies that $(c_{n-1},c_0,c_1,\ldots,c_{n-2})\in\mc{C}$. Cyclic codes are an important class of linear codes and have many applications in engineering and coding theory. Some recent developments of cyclic codes can be found in \cite{BS06,DL13,DY13,FL08,Feng12,JLX11,LF08,ZHJYC10,ZSH12} and the references therein.

Let $\gcd(n,p)=1$. By identifying any codeword $(c_0,c_1,\ldots,c_{n-1})\in \mc{C}$ with
\[c_0+c_1x+c_2x^2+\cdots+c_{n-1}x^{n-1}\in\F_p[x]/(x^n-1),\]
any code of length $n$ over $\F_p$ corresponds to an ideal of the polynomial residue class ring $\F_p[x]/(x^n-1)$. It is well-known that every ideal of $\F_p[x]/(x^n-1)$ is principal. Thus a cyclic code $\mc{C}$ can be expressed as $\mc{C}=\<g(x)\>$ where $g(x)$ is monic and has the least degree. The polynomial $g(x)$ is called a {\it generator polynomial} of $\mc{C}$ and $h(x):=(x^n-1)/g(x)$ is called the {\it check polynomial} of $\mc{C}$.

Let $\alpha$ be a generator of the multiplicative group $\F_{p^m}^*$ of $\F_{p^m}$. For any $0\leq i\leq p^m-2$, denote by $m_i(x)$ the minimal polynomial of $\alpha^i$ over $\F_p$. Let $\mc{C}_p(1,e)$ be the $p$-ary cyclic code of length $p^m-1$ over $\F_p$ with generator polynomial $m_1(x)m_{e}(x)$ where $1<e<p^m-1$ and $e$ is not a power of $p$. It is known that $\mc{C}_p(1,e)$ has dimension $p^m-1-m-\deg(m_e(x))$. Over the years, there has been a lot of interest in the case where $\deg(m_e(x))=m$ and so the code $\mc{C}_p(1,e)$ has dimension $p^m-1-2m$. 
When $p=2$, van Lint and Wilson \cite{VLW86} showed that the minimum distance of $\mc{C}_2(1,e)$ is at most 5 when $m\geq 4$. Thus a binary code $\mc{C}_2(1,e)$ with parameters $[2^m-1,2^m-1-2m,5]$ is considered to be optimal in this sense. Regarding optimal binary codes $\mc{C}_2(1,e)$, it was proved in \cite{CCD00,CCZ98} that $\mc{C}_2(1,e)$ has parameters $[2^m-1,2^m-1-2m,5]$ if and only if $x^e$ is an almost perfect nonlinear monomial over $\F_{2^m}$. 
When $p>3$, as pointed out in \cite{DH13}, a $[p^m-1,p^m-1-2m]$ code $\mc{C}_p(1,e)$ has minimum distance 2 or 3, which might not be interesting. 
When $p=3$, any ternary code with parameters $[3^m-1,3^m-1-2m,4]$ is optimal with respect to the sphere packing bound \cite{HP03}. In 2005, Carlet, Ding and Yuan \cite{CDY05} constructed optimal ternary cyclic codes $\mc{C}_3(1,e)$ with parameters $[3^m-1,3^m-1-2m,4]$ by using perfect nonlinear monomials $x^e$. In 2013, Ding and Helleseth \cite{DH13} constructed several new classes of optimal ternary cyclic codes  $\mc{C}_3(1,e)$ with the same parameters by using almost perfect nonlinear monomials, and a number of other monomials over $\F_{3^m}$. Moreover, nine open problems on this kind of optimal ternary cyclic codes were proposed in the same paper. Since then, optimal ternary cyclic codes with parameters $[3^m-1,3^m-1-2m,4]$ have been studied extensively. Many such codes have been constructed and some of the open problems of Ding and Helleseth have been solved. To the best of our knowledge, the known constructions of optimal ternary cyclic codes $\mc{C}_3(1,e)$ are summarized in Table \ref{tab_known}. For brevity and consistency with the literature, we will use $\mc{C}_{(1,e)}$ to denote $\mc{C}_3(1,e)$ throughout this paper.

\begin{table}[h]\aboverulesep=0pt \belowrulesep=0pt
\setlength{\abovecaptionskip}{0cm}
\setlength{\belowcaptionskip}{0cm}
\caption{\footnotesize Known ternary cyclic codes $\mc{C}_{(1,e)}$ with parameters $[3^m-1,3^m-1-2m,4]$}
\label{tab_known}
\centering
\[\footnotesize
\arraycolsep=5pt
\begin{tabular}{llc}
\toprule
$e$ ($e$ is even) & \text{Conditions} &  \text{Reference} \\
\midrule
&&\\[-0.4cm]
$2$ & $m\geq 2$  & \cite{DH13}\\
$16$ & $m$ is odd, $3\nmid m$ & \cite{LLHDT14}\\
$20$ & $m$ is odd & \cite{LLHDT14}\\
$3^s+1$ &$m/\gcd(m,s)$ is odd & \cite{DO68}\\
$3^s-1$ & $\gcd(m,s)=\gcd(3^m-1,3^s-2)=1$ & \cite{DH13}\\
$(3^s+1)/2$ & $\gcd(m,s)=1$ and $s$ is odd & \cite{CM97}\\
$(3^s-1)/2$ & $m$ is odd, $s$ is even, $\gcd(m,s)=\gcd(m,s-1)=1$ &\cite{DH13}\\
$(3^s+7)/2$ & $m$ is odd and $s$ is even &\cite{ZH20}\\
$(3^s+7)/2+(3^m-1)/2$ & $m$ is odd and $s$ is odd &\cite{ZH20}\\
$2(3^s+1)$ & $m$ is odd & \cite{LZH15}\\
$3^{m/2}+5$ & $m\geq 4$, $m\equiv 0\pmod{4}$ &\cite{HY19}\\
$3^{(m+1)/2}-1$ & $m$ is odd & \cite{DH13}\\
$3^{(m+2)/2}+5$ & $m\geq 6$, $m\equiv 2\pmod{4}$ &\cite{HY19}\\
$(3^{(m+1)/2}-1)/2, (3^{m+1}-1)/8$ & $m\equiv 3\pmod{4}$ &\cite{DH13}\\
$(3^{(m+1)/2}-1)/2+(3^m-1)/2$& $m\equiv 1\pmod{4}$ & \cite{DH13} \\
$(3^{(m+1)/2}+5)/2$ & $m\equiv 1\pmod{4}$ &\cite{ZH20}\\
$(3^{(m+1)/2}+5)/2+(3^m-1)/2$ & $m\equiv 3\pmod{4}$ & \cite{ZH20}\\
$(3^{(m+1)/4}-1)(3^{(m+1)/2}+1)$ & $m\equiv 3\pmod{4}$ &\cite{DH13}\\
$3^{m-1}-1$, $(3^m+1)/4+(3^m-1)/2$ & $m\geq 3$, $m$ is odd &\cite{DH13}\\
$2(3^{m-1}-1)$, $5(3^{m-1}-1)$ & $m$ is odd, $3\nmid m$ & \cite{LLHDT14}\\
$(3^m-3)/2$ & $m\geq 5$, $m$ is odd & \cite{DH13}\\
$(3^m-3)/4$ & $m$ is odd &\cite{YH19}\\
$(3^m-1)/2-2$, $(3^m-1)/2+10$ & $m\equiv 2\pmod{4}$ & \cite{LLHDT14}\\
$(3^m-1)/2-5$, $(3^m-1)/2+7$ & $m$ is odd & \cite{LLHDT14}\\
$(3^{m+1}-1)/8+(3^m-1)/2$ & $m\equiv 1\pmod{4}$ & \cite{DH13}\\
$e\equiv(3^{m-1}-1)/2+3^h\pmod{3^m-1}$ & $m$ is even, $\gcd(m,h+1)=\gcd(3^{h+1}-2,3^m-1)=1$ &\cite{ZLS22}\\
$e\equiv (3^m-1)/2+3^s+1\pmod{3^m-1}$ & $m$ is even, $m/\gcd(m,s)$ is odd & \cite{WW16}\\
$e\equiv (3^m-1)/2+3^s-1\pmod{3^m-1}$  & $m$ is even, $\gcd(m,s)=\gcd(3^m-1,3^s-2)=1$& \cite{WW16}\\
$5e\equiv 2\pmod{3^m-1}$ & $3\nmid m$ &\cite{ZH20}\\
$7e\equiv 2\pmod{3^m-1}$ & $5\nmid m$ &\cite{ZH20}\\
$5e\equiv 4\pmod{3^m-1}$ & $m>2$, $3\nmid m$, $5\nmid m$ & \cite{ZH20}\\
$e(3^s+1)\equiv 3^t+1\pmod{3^m-1}$ & $\gcd(m,t-s)=\gcd(m,t+s)=1$ &\cite{WW16,ZLS22}\\
$e(3^s+1)\equiv (3^m+1)/2\pmod{3^m-1}$ & $m$ is odd & \cite{ZH20}\\
$e(3^s-1)\equiv (3^t-1) \pmod{3^m-1}$ & $\gcd(m,t)=\gcd(m,t-s)=1$  &\cite{WW16,ZLS22}\\
\bottomrule
\end{tabular}
\]
\end{table}

In this paper, we investigate optimal ternary cyclic codes $\mathcal{C}_{(1,e)}$ with parameters $[3^m-1, 3^m-1-2m, 4]$. Our work consists of two main parts. 
First, we address two open problems, i.e., Problem \ref{prob7.9} and Problem \ref{prob7.10}, of Ding and Helleseth proposed in \cite{DH13}. We present the first counterexamples to both problems and solve two special cases of Problem \ref{prob7.9} by constructing two families of optimal codes that satisfy the conditions of Problem \ref{prob7.9}. 
Second, we conduct a comprehensive study of cyclic codes $\mc{C}_{(1,e)}$ satisfying the congruence relation $e(3^h\pm 1) \equiv \frac{3^m-a}{2} \pmod{3^m-1}$ for odd $a$. We obtain two new families of optimal codes $\mc{C}_{(1,e)}$ with $e$ in this form, and several nonexistence results on cyclic codes $\mc{C}_{(1,e)}$ with parameters $[3^m-1,3^m-1-2m,4]$.

The paper is organized as follows. Section \ref{sec_prem} provides necessary background. Section \ref{sec_DH} presents our contributions to Ding and Helleseth's problems. Section \ref{sec_newcodes} constructs new optimal codes $\mc{C}_{(1,e)}$ with $e(3^h\pm 1) \equiv \frac{3^m-a}{2} \pmod{3^m-1}$, while Section \ref{sec_nonexist} establishes nonexistence results. We conclude in Section \ref{sec_conclude} with a summary and discussion of codes $\mc{C}_{(1,e)}$ satisfying $e(3^h\pm 1) \equiv \frac{3^m-a}{2} \pmod{3^m-1}$ and $a\equiv 1\pmod{8}$.

\section{Preliminaries}\label{sec_prem}

In this section, we introduce some tools for constructing optimal ternary cyclic codes with parameters $[3^m-1,3^m-1-2m,4]$.

For a prime $p$, the $p$-cyclotomic coset modulo $p^m-1$ containing $j$ is defined as
\[C_j=\{j\cdot p^s\pmod{p^m-1}: s=0,1,\ldots,m-1\}.\]

In \cite{DH13}, Ding and Helleseth provided the following result giving a way to determine whether a ternary cyclic code $\mc{C}_{(1,e)}$ has parameters $[3^m-1,3^m-1-2m,4]$ or not. It will be our major tool to construct new optimal codes.

\begin{lemma}[{\cite[Theorem 4.1]{DH13}}]\label{lem_dingconds}
Let $e\notin C_1$ and $|C_e|=m$. The ternary cyclic code $\mc{C}_{(1,e)}$ has parameters $[3^m-1,3^m-1-2m,4]$ if and only if the following conditions are satisfied:

{\rm C1}: $e$ is even;

{\rm C2}: $(x+1)^e+x^e+1=0$ has a unique solution $x=1$ over $\F_{3^m}$; and

{\rm C3}: $(x+1)^e-x^e-1=0$ has a unique solution $x=0$ over $\F_{3^m}$.
\end{lemma}

\begin{remark}\label{rmk_nescond}
In order to construct ternary cyclic code $\mc{C}_{(1,e)}$ with parameters $[3^m-1,3^m-1-2m,4]$, we must choose $e$ to be even and satisfy that $|C_e|=m$. The reason is as follows. First of all, it is well-known that the dimension of $\mc{C}_{(1,e)}$ is equal to $3^m-1-m-|C_e|$. So the dimension of $\mc{C}_{(1,e)}$ equals $3^m-1-2m$ if and only if $|C_e|=m$. Moreover, if $e$ is odd and $|C_e|=m$, then Lemma \ref{lem_dingconds} guarantees that the minimum distance of $\mc{C}_{(1,e)}$ is not $4$.
\end{remark}

The following lemma allows us to confirm that $|C_e|=m$ in some situation.

\begin{lemma}[{\cite[Lemma 2.1]{DH13}}]\label{lem_Cesize}
For any $1\leq e\leq p^m-2$ with $\gcd(e,p^m-1)=2$, the size of the $p$-cyclotomic coset $C_e$ is equal to $m$.
\end{lemma}

\begin{lemma}[{\cite[Corollary 3.47]{LN97}}]\label{lem_irr}
 An irreducible polynomial over $\mathbb{F}_{p^m}$ of degree $n$ remains irreducible over $\mathbb{F}_{p^{ml}}$ if and only if $\gcd(l,n)=1$.
\end{lemma}
 
Let $n$ be a nonzero integer and $p$ a prime. We use $(n)_p$ to denote the {\it $p$-part} of $n$ which is the highest power of $p$ dividing $n$. We use $(n)_{p'}$ to denote the {\it $p'$-part} of $n$ which is defined as $\frac{n}{(n)_p}$. For example, we have $(12)_2=4$ and $(12)_{2'}=3$.

The following lemma is basic and will be used frequently throughout the paper. We include a proof of it for completeness.
\begin{lemma}\label{lem_gcd}
Let $q>1$, $k, \ell$ be positive integers. Then 
\begin{enumerate}
\item[\rm(1)] $\gcd(q^k-1,q^\ell-1)=q^{\gcd(k,\ell)}-1$.
\item[\rm(2)] $\gcd(q^k+1,q^\ell-1)=\left\{\begin{array}{ll}
\vspace{0.1cm}q^{\gcd(k,\ell)}+1, &\text{if }(k)_2<(\ell)_2,\\
\gcd(2,q+1),&\text{otherwise}.
\end{array}
\right .
$
\end{enumerate}

\end{lemma}
\begin{proof}
Part (1) can be easily derived by iteratively applying the Euclidean algorithm. More precisely, if $k=\ell s+t$ with $0\leq t<\ell$, then 
\[\gcd(q^k-1,q^\ell-1)=\gcd(q^t(q^{\ell s}-1)+q^t-1,q^\ell-1)=\gcd(q^t-1,q^\ell-1).\] 
Repeat this process until the remainder becomes 0, so we obtain
\[\gcd(q^k-1,q^\ell-1)=\gcd(q^0-1,q^{\gcd(k,\ell)}-1)=q^{\gcd(k,\ell)}-1.\]
 
Next, we prove part (2). The proof is divided into the following two cases:

{\bf Case 1}. $(k)_2<(\ell)_2$. 

In this case, $\ell$ is even. Let $d=\gcd(k,\ell)$. Then $d=ku+\ell v$ for some integers $u$ and $v$. This equation can be rewritten as
\[(d)_2(d)_{2'}-(k)_2(k)_{2'}u=\ell v.\]
Note that $(d)_2=(k)_2<(\ell)_2$. Then $(d)_{2'}-(k)_{2'}u=(\ell)_2(\ell)_{2'}v/(d)_2$ is even. So $u$ must be odd. 

Let $D=\gcd(q^k+1,q^\ell-1)$. Then we have $q^k\equiv -1\pmod{D}$ and $q^\ell\equiv 1\pmod{D}$. Thus, 
\[q^d+1=(q^k)^u(q^\ell)^v+1\equiv (-1)^u1^v+1\equiv -1+1\equiv 0\pmod{D}.\]
On the other hand, since $d=\gcd(k,\ell)$, we have $k=dy$ and $\ell=dz$ for some integers $y$ and $z$. Since $(d)_2=(k)_2$, then $y$ is odd. Thus, we deduce that 
\[(q^d+1)\mid (q^d)^y+1=q^{k}+1.\] 
Similarly, since $(d)_2<(\ell)_2$, then $z$ is even. Therefore, 
\[(q^d+1)\mid (q^{2d}-1)\mid q^{zd}-1=q^\ell-1\]
Combining both results, we conclude that $q^d+1$ divides both $q^{k}+1$ and $q^\ell-1$, hence $(q^d+1)\mid D$. This completes the proof for the case when $(k)_2<(\ell)_2$.

{\bf Case 2}. $(k)_2\geq (\ell)_2$. 

In this case, we observe that $\gcd(k,\ell)=\gcd(2k,\ell)$. Consequently, 
\begin{equation}\label{eq_lemgcd1}
\gcd((q^k-1)(q^k+1),q^\ell-1)=q^{\gcd(2k,\ell)}-1=q^{\gcd(k,\ell)}-1=\gcd(q^k-1,q^\ell-1).
\end{equation}
If $q$ is even, then $\gcd(q^k-1,q^k+1)=\gcd(q^k-1,2)=1$. It follows from \eqref{eq_lemgcd1} that 
\[ \gcd(q^k+1,q^\ell-1)=1=\gcd(2,q+1).\] 
If $q$ is odd, then $\gcd(2,q+1)=2$. Let $a=q^k-1, b=q^k+1$ and $c=q^\ell-1$. Note that 
\[\gcd(b,c)=\gcd((b)_2,(c)_2)\cdot\gcd((b)_{2'},(c)_{2'}).\] 
In order to prove that $\gcd(b,c)=2$, we only need to show that 
\[ \gcd((b)_2,(c)_2)=2\ {\rm and}\ \gcd((b)_{2'},(c)_{2'})=1.\] 
By \eqref{eq_lemgcd1}, we have $\gcd((a)_{2'}(b)_{2'},(c)_{2'})=\gcd((a)_{2'},(c)_{2'})$. Since  $\gcd(a,b)=\gcd(a,2)=2$, it follows that $\gcd((a)_{2'},(b)_{2'})=1$. Thus, 
\[ \gcd((b)_{2'},(c)_{2'})=1.\] 
Since $q$ is odd, we have $(a)_2\geq 2, (b)_2\geq 2$ and $(c)_2\geq 2$. We now consider two subcases for $(b)_2$:

If $(b)_2=2$, then it is clear that $\gcd((b)_2,(c)_2)=2$. 

If $(b)_2\geq 4$, then we have $(a)_2=2$ since $\gcd(a,b)=2$. By \eqref{eq_lemgcd1}, we have
 \[ \gcd((a)_{2}(b)_{2},(c)_{2})=\gcd((a)_{2},(c)_{2}).\]
This implies that $2=(a)_2\geq (c)_2\geq 2$. Therefore, $(c)_2=2$ and $\gcd((b)_2,(c)_2)=2$. This completes the proof. 
\end{proof}

%

Finally, we recall the definition and some simple properties of the quadratic character of a finite field. Let $q$ be an odd prime power. The quadratic character of $\F_q^*$ is the function $\eta: \F_q^*\to \{\pm 1\}$ defined by $\eta(x)=1$ if $x$ is a nonzero square in $\F_q^*$ and $\eta(x)=-1$ if $x$ is a nonsquare in $\F_q^*$. For example, we have $\eta(1)=1$ and $\eta(\alpha)=-1$ where $\alpha$ is a generator of $\F_q^*$. It is well-known that $\eta(-1)=-1$ if and only if $q\equiv 3\pmod{4}$. 

We can and will extend the definition of $\eta$ to $\F_q$ by defining $\eta(0)=0$. Then
\[\eta(x)=x^{\frac{q-1}{2}}, \text{for all } x\in\F_q.\]

We will call $\eta$ the quadratic character of $\F_q$ in the rest of the paper.

\section{Ding and Helleseth's open problems and the first two new families of optimal ternary cyclic codes}\label{sec_DH}
In this section, we are mainly concerned with the following problems proposed by Ding and Helleseth in \cite{DH13}.
\begin{problem}[{\cite[Open Problem 7.9]{DH13}}]\label{prob7.9}
Let $e=(3^h+5)/2$, where $1\leq h\leq m-1$. Is it true that the ternary cyclic code $\mc{C}_{(1,e)}$ has parameters $[3^m-1,3^m-1-2m,4]$ if

1) $m$ is odd and $m\not\equiv 0\pmod{3}$;

2) $h$ is odd; and

3) $\gcd(h,m)=1$?
\end{problem}

\begin{problem}[{\cite[Open Problem 7.10]{DH13}}]\label{prob7.10}
Let $e=(3^h-5)/2$, where $2\leq h\leq m-1$. Let $m$ be odd and $m\not\equiv 0\pmod{3}$. Is it true that the ternary code $\mc{C}_{(1,e)}$ has parameters $[3^m-1,3^m-1-2m,4]$ if $h$ is even?
\end{problem}

Before our work, there was no known answer to Problem \ref{prob7.9}, and the only known answer to Problem \ref{prob7.10} was given in \cite{LLHDT14}, where the authors proved that $\mc{C}_{(1,e)}$ has parameters $[3^m-1,3^m-1-2m,4]$ if $e=\frac{3^{m-1}-5}{2}$. In the following, we first present some counterexamples to both problems. Then, we show that in other cases the answer to Problem \ref{prob7.9} is positive. This is done by constructing two families of optimal cyclic codes that satisfy all the conditions of Problem \ref{prob7.9}.

\subsection{Counterexamples to Problem \ref{prob7.9}}
We present two counterexamples to Problem \ref{prob7.9}.
\begin{example}
Let  $m=245$ and $h=9$. Then $e=(3^h+5)/2=9844$. By Magma \cite{Magma}, we can decompose $(x+1)^e-x^e-1$ into a product of  irreducible polynomials over $\mathbb{F}_3$:
\begin{align*}
&(x+1)^e-x^e-1=f(x)\cdot m(x),
\end{align*}
where
{\footnotesize
\begin{align*}
f(x) = &\ x^{245} -x^{244} -x^{243} + x^{242} -x^{240} -x^{239} + x^{238} + x^{235} + x^{234} + x^{233} -x^{232} -x^{231}-x^{230}-x^{228} + x^{227} \\ 
& -x^{226} -x^{224} -x^{221} + x^{220} + x^{219} + x^{218} -x^{216} -x^{215} + x^{214} -x^{210} + x^{208}+ x^{207} -x^{205} -x^{204} \\
& + x^{202} + x^{201} -x^{200} + x^{199} -x^{197} -x^{196} + x^{193} -x^{191} + x^{190} + x^{189}+ x^{188} -x^{187} + x^{186} + x^{185} + x^{183} \\
& + x^{182} + x^{180} + x^{178} + x^{177}-x^{176} + x^{175} -x^{173} + x^{172}-x^{170} -x^{169} + x^{168} + x^{167} + x^{166} -x^{165} + x^{164} \\
&+ x^{162}+ x^{161} -x^{159} + x^{158} + x^{156} + x^{155}-x^{153} + x^{152} + x^{148} + x^{145} + x^{144} -x^{143} -x^{142} -x^{141} -x^{140} \\
& -x^{138} +x^{134} -x^{131} + x^{130} -x^{129} -x^{126} -x^{125} -x^{124} + x^{123} + x^{122}+ x^{120} + x^{119} + x^{118} + x^{117}+ x^{116}\\
&+ x^{115} -x^{113} + x^{112} + x^{110} -x^{108} + x^{104} -x^{103}+ x^{102} -x^{100} + x^{99} -x^{98} + x^{97} -x^{96} + x^{94} -x^{92} \\
&+ x^{90} + x^{87} -x^{85} + x^{83}-x^{81} -x^{79} + x^{78} -x^{77} + x^{74} + x^{72} -x^{71} -x^{70}+ x^{67}-x^{66} + x^{65}+ x^{62}-x^{61}\\
& -x^{59} -x^{58} -x^{55} + x^{54} + x^{53} + x^{51} + x^{50} + x^{48} -x^{47} + x^{46} -x^{45}+ x^{44}+ x^{43} + x^{42} + x^{41} + x^{39} + x^{38}\\
& -x^{36} + x^{35} -x^{34} -x^{33} -x^{32}-x^{31} -x^{30} + x^{29} + x^{28}+ x^{27} + x^{25} -x^{23} + x^{22} -x^{21} -x^{20} + x^{18} -x^{17} \\
& -x^{14} -x^{12}-x^{11} + x^{10} -x^{9} + x^{3} + x + 1
\end{align*}}
is an irreducible polynomial of degree $245$ over $\mathbb{F}_3$, and $m(x)$ is the product of other irreducible polynomials over $\mathbb{F}_3$. Thus $(x+1)^e-x^e-1=0$ has at least $245$ solutions in $\mathbb{F}_{3^m}$ by Lemma \ref{lem_irr}. From Lemma \ref{lem_dingconds} C3, we know that $\mc{C}_{(1,e)}$ does not have parameters $[3^m-1,3^m-1-2m,4]$.
\end{example}

\begin{example}
Let $m=109$ and $h=101$. Then 
\[e=(3^h+5)/2=773066281098016996554691694648431909053161283004.\] Let 
\begin{align*}
f_1(x) = &\ x^{109} - x^{106} + x^{105} + x^{102} - x^{101} + x^{100} - x^{97} - x^{96} - x^{95} - x^{94} - x^{92}- x^{88} \\ 
         & - x^{87} - x^{85} + x^{84} - x^{81} - x^{80} + x^{79}- x^{78} - x^{76} + x^{75} + x^{73} + x^{72} - x^{68}\\
         &  + x^{67} + x^{66} + x^{65}- x^{62} - x^{61} - x^{60} - x^{59} + x^{58} + x^{57}  - x^{56} + x^{55} + x^{53}\\
         & - x^{52}+ x^{51} - x^{49} + x^{47} + x^{46} + x^{43} - x^{41} + x^{39}- x^{37} + x^{36} - x^{35} - x^{34}\\
         & - x^{33} + x^{32}  - x^{30} - x^{28} - x^{27} + x^{26}- x^{25} + x^{23} + x^{22} + x^{19} - x^{18} + x^{17}\\ 
         & - x^{15} + x^{13}+ x^{12}- x^{11} + x^{10} + x^8 + x^7 + x^6 - x^4 + x^3 + 1,\\
f_2(x) = &\ x^{109} + x^{106} - x^{105} + x^{103} + x^{102} + x^{101} + x^{99} - x^{98} + x^{97} + x^{96}- x^{94} + x^{92} \\
      & - x^{91}+ x^{90} + x^{87} + x^{86} - x^{84} + x^{83} - x^{82} - x^{81} - x^{79} + x^{77} - x^{76} - x^{75} \\
      & - x^{74} + x^{73}- x^{72} + x^{70} - x^{68} + x^{66}+ x^{63} + x^{62} - x^{60} + x^{58} - x^{57}  + x^{56}\\
      & + x^{54} - x^{53}+ x^{52}+ x^{51} - x^{50} - x^{49} - x^{48} - x^{47} + x^{44} + x^{43} + x^{42} - x^{41} \\
      &+ x^{37} + x^{36} + x^{34} - x^{33} - x^{31} + x^{30} - x^{29} - x^{28} + x^{25} - x^{24} - x^{22} - x^{21} \\
      &- x^{17} - x^{15} - x^{14} - x^{13} - x^{12} + x^9 - x^8 + x^7 + x^4 - x^3 + 1.
\end{align*}
We check by Magma \cite{Magma} that both $f_1(x)$ and $f_2(x)$ are irreducible over $\F_3$, and all their roots satisfy $(x+1)^e-x^e-1=0$. Thus $(x+1)^e-x^e-1=0$ has at least $218$ solutions in $\mathbb{F}_{3^m}$ by Lemma \ref{lem_irr}. From Lemma \ref{lem_dingconds} C3, we know that $\mc{C}_{(1,e)}$ does not have parameters $[3^m-1,3^m-1-2m,4]$.
\end{example}

\subsection{Counterexamples to Problem \ref{prob7.10}}
We present two counterexamples to Problem \ref{prob7.10}.
\begin{example}
Let $m=295$ and $h=10$. Then $e=(3^h-5)/2=29522$. By Magma \cite{Magma}, we can decompose $(x+1)^e+x^e+1$ into a product of irreducible polynomials over $\mathbb{F}_3$:
\begin{align*}
&(x+1)^e+x^e+1=f(x)\cdot m(x),
\end{align*}
where 
{\footnotesize
\begin{align*}
f(x)=&x^{295} -x^{294} + x^{293} -x^{292} + x^{291} -x^{289} -x^{288} -x^{286} -x^{284} -x^{282} -x^{281} + x^{280} -x^{279} + x^{278} -x^{277}\\
& -x^{275} + x^{274} -x^{273} + x^{270} + x^{268} -x^{266} + x^{264} + x^{263} -x^{262} -x^{261} -x^{259} + x^{258} + x^{257} -x^{256} + x^{255} \\
& + x^{254} -x^{253} -x^{251} -x^{250} -x^{248} -x^{247} -x^{246} + x^{241} -x^{240} + x^{237} -x^{232} -x^{231} -x^{229} + x^{228} -x^{227}\\
& -x^{226} -x^{223} + x^{220} -x^{218} -x^{217} + x^{216} -x^{215} + x^{214} + x^{213} + x^{212} + x^{211} -x^{209} + x^{208} + x^{207} -x^{206} \\
&-x^{205} -x^{204} -x^{203} -x^{202} + x^{201} + x^{200} -x^{195} + x^{193} -x^{191} + x^{189} + x^{188} -x^{187} 
+ x^{186} + x^{184} -x^{183}\\
& -x^{180} + x^{177} -x^{176} -x^{174} -x^{170} + x^{169} + x^{164} -x^{163} + x^{162} + x^{160} + x^{158} -x^{157} -x^{156} -x^{155} + x^{153}\\
&-x^{152} -x^{151} + x^{150} + x^{148} -x^{146} + x^{144} + x^{142} -x^{140} + x^{137} + x^{135} -x^{130} + x^{128} + x^{126} + x^{123} -x^{121}\\
&-x^{120} + x^{119} + x^{118}+ x^{116} -x^{113} + x^{111} + x^{110} -x^{108} + x^{104} -x^{102} + x^{101} -x^{99} -x^{97} + x^{94} -x^{93}\\
& -x^{92} + x^{91} + x^{90} + x^{89} -x^{85} + x^{84} -x^{82} -x^{81} + x^{80} -x^{79} -x^{78} + x^{77} -x^{75} + x^{74} + x^{73}+ x^{71} -x^{70}\\
& -x^{69} -x^{68} -x^{67} -x^{64} + x^{58} -x^{57} + x^{56} + x^{53} + x^{50} -x^{49} -x^{48} + x^{47} 
-x^{46} -x^{44} -x^{43} -x^{40}+ x^{37}\\
& -x^{36} -x^{35} -x^{34} -x^{33} + x^{32} + x^{31} -x^{27} + x^{21}-x^{20} -x^{18} -x^{17} + x^{15} + x^{14} + x^{12} + x^{10}\\
& + x^{9} + x^{8} + x^{7} + x^{6} -x^{5} -x^{3} + x^{2} + 2 \\
\end{align*}}
is an irreducible polynomial of degree $295$ over $\mathbb{F}_3$, and $m(x)$ is the product of other irreducible polynomials over $\mathbb{F}_3$. Thus $(x+1)^e+x^e+1=0$ has at least $295$ solutions in $\mathbb{F}_{3^m}$ by Lemma \ref{lem_irr}. From Lemma \ref{lem_dingconds} C2, we know that $\mc{C}_{(1,e)}$ does not have parameters $[3^m-1,3^m-1-2m,4]$.
\end{example}

\begin{example}
Let $m=35$, $h=32$. Then $e=(3^h-5)/2=926510094425918$. Let
\begin{align*}
f_1(x) &= x^{35} - x^{32} + x^{31} + x^{30} + x^{29} + x^{28} - x^{26} - x^{25} - x^{24}  + x^{22} - x^{18} \\
&\quad- x^{14} + x^{13} - x^{12} - x^{11} + x^{10} + x^9 - x^8 + x^7 + x^6 + x^4 - x^3 - 1, \\
f_2(x) &= x^{35} + x^{32} - x^{31} - x^{29} - x^{28} + x^{27} - x^{26} - x^{25} + x^{24}+ x^{23} - x^{22} \\
       &\quad  + x^{21} + x^{17} - x^{13} + x^{11} + x^{10} + x^9 - x^7 - x^6 - x^5 - x^4 + x^3 - 1.
\end{align*}
We check by Magma \cite{Magma} that both $f_1(x)$ and $f_2(x)$ are irreducible over $\F_3$, and all their roots satisfy $(x+1)^e-x^e-1=0$. Thus $(x+1)^e-x^e-1=0$ has at least $70$ solutions in $\mathbb{F}_{3^m}$ by Lemma \ref{lem_irr}. From Lemma \ref{lem_dingconds} C3, we know that $\mc{C}_{(1,e)}$ does not have parameters $[3^m-1,3^m-1-2m,4]$.
\end{example}

\subsection{Two families of optimal ternary cyclic codes with minimum distance four that satisfy all the conditions of Problem \ref{prob7.9}.}
Although we have found counterexamples to Problem \ref{prob7.9}, the answer to Problem \ref{prob7.9} is not always negative. In this subsection, we present two families of optimal cyclic codes that give positive answers to Problem \ref{prob7.9} in different cases. 

In our first family, we assume that $m\equiv 7, 11 \pmod{12}$ and $h=\frac{m+3}{2}$.

\begin{lemma}\label{lem_main0}
Let $m$ be an integer such that $m\equiv 7, 11 \pmod{12}$ and let $e=\frac{3^{\frac{m+3}{2}}+5}{2}$. Then $e\notin C_1$ and $|C_e|=m$.
\end{lemma}

\begin{proof}
Since $m\equiv 7, 11 \pmod{12}$, we have $\frac{m+3}{2}$ is odd. Then $3^{\frac{m+3}{2}}\equiv 3\pmod{4}$. Thus, $e$ is even and so $e\notin C_1$. Next, we show that $|C_e|=m$. Since $m$ is odd, we have 
\begin{align*}
&\gcd(e, 3^m-1)=\gcd\left(\frac{3^{\frac{m+3}{2}}+5}{2},3^m-1\right)\\
=&\gcd\left(3^{\frac{m+3}{2}}+5,3^m-1\right)=\gcd\left(3^{\frac{m+3}{2}}+5,3^{m+3}-3^3\right)\\
=&\gcd\left(3^{\frac{m+3}{2}}+5,3^{m+3}-3^3-(3^{\frac{m+3}{2}}+5)(3^{\frac{m+3}{2}}-5)\right)\\
=&\gcd\left(3^{\frac{m+3}{2}}+5,2\right)=2.
\end{align*}
Consequently, by Lemma \ref{lem_Cesize}, we have $|C_e|=m$. This completes the proof.
\end{proof}

After proving Lemma \ref{lem_main0}, we now consider conditions C2 and C3 in Lemma \ref{lem_dingconds} for $e=\frac{3^{\frac{m+3}{2}}+5}{2}$. 

\begin{lemma}\label{lem_main1}
Let $m$ be an integer such that $m\equiv 7, 11 \pmod{12}$ and let $e=\frac{3^{\frac{m+3}{2}}+5}{2}$. Then the equation 
\begin{equation}\label{eq01}
(x+1)^{2e}=(x^e+1)^2
\end{equation}
has no solutions in $\F_{3^m}\setminus\F_3$ if $\gcd(m, 13)=1.$
\end{lemma}

\begin{proof}
It is easy to see that $x=0, 1$ are the only solutions in $\mathbb{F}_3$ for (\ref{eq01}). It can be verified
that 
\begin{align*}
&(x+1)^{2e}=(x+1)^{3^{\frac{m+3}{2}}+5}=(x+1)^{3^{\frac{m+3}{2}}}(x+1)^3(x+1)^2\\
=&x^{3^{\frac{m+3}{2}}+5}-x^{3^{\frac{m+3}{2}}+4}+x^{3^{\frac{m+3}{2}}+3}+x^{3^{\frac{m+3}{2}}+2}-x^{3^{\frac{m+3}{2}}+1}+x^{3^{\frac{m+3}{2}}}+x^5-x^4+x^3+x^2-x+1.
\end{align*}
Note that (\ref{eq01}) can be simplified as
\begin{equation}\label{eq02}
(x^4-x^3-x^2+x-1)(x^{3^{\frac{m+3}{2}}-1}-1)=x^2(x^{\frac{3^{\frac{m+3}{2}}-1}{2}}-1).
\end{equation}
Suppose that $\theta\in \mathbb{F}_{3^m}\setminus \mathbb{F}_3$ is a solution for (\ref{eq01}), then we have
\begin{equation}\label{eq02}
(\theta^4-\theta^3-\theta^2+\theta-1)(\theta^{3^{\frac{m+3}{2}}-1}-1)=\theta^2(\theta^{\frac{3^{\frac{m+3}{2}}-1}{2}}-1).
\end{equation}
If $\theta^{\frac{3^{\frac{m+3}{2}}-1}{2}}=1$, then we have $\theta^{3^{\frac{m+3}{2}}-1}=1$. Since $\gcd(m,3)=1$, we have $\gcd(\frac{m+3}{2}, m)=\gcd(m+3, m)=1$. Then 
$$1=\theta^{\gcd(3^{\frac{m+3}{2}}-1, 3^m-1)}=\theta^{3^{\gcd(\frac{m+3}{2}, m)}-1}=\theta^2,$$
which implies that $\theta\in \mathbb{F}_3$. This contradicts the assumption that $\theta\in \mathbb{F}_{3^m}\setminus \mathbb{F}_3$. Hence we have $\theta^{\frac{3^{\frac{m+3}{2}}-1}{2}}\neq 1$. We then deduce from \eqref{eq02} that
\begin{equation}\label{eq03}
\theta^{3^{\frac{m+3}{2}}}=\frac{\theta(-\theta^4 + \theta^3 - \theta^2 -\theta + 1)^2}
{(\theta^4  -\theta^3 -\theta^2 + \theta - 1)^2}\triangleq\frac{f(\theta)}{g(\theta)},
\end{equation}
where $f(\theta)=\theta(-\theta^4 + \theta^3 - \theta^2 -\theta + 1)^2$ and $g(\theta)=(\theta^4  -\theta^3 -\theta^2 + \theta - 1)^2$. Raising both sides of Equation \eqref{eq03} to the $3^{\frac{m+3}{2}}$-th power, we have
\begin{equation}\label{eq04}
\theta^{3^{m+3}}=\frac{\theta^{3^{\frac{m+3}{2}}}(-\theta^{4\cdot 3^{\frac{m+3}{2}}}+\theta^{3\cdot 3^{\frac{m+3}{2}}}
-\theta^{2\cdot 3^{\frac{m+3}{2}}}-\theta^{3^{\frac{m+3}{2}}}+1)^2}{(\theta^{4\cdot 3^{\frac{m+3}{2}}}-\theta^{3\cdot 3^{\frac{m+3}{2}}}-\theta^{2\cdot 3^{\frac{m+3}{2}}}+\theta^{3^{\frac{m+3}{2}}}-1)^2}.
\end{equation}
Plugging $\theta^{3^{\frac{m+3}{2}}}=\frac{f(\theta)}{g(\theta)}$ into \eqref{eq04}, we obtain
\begin{equation}\label{eq05}
\theta^{3^{m+3}}=\frac{F(\theta)}{G(\theta)},
\end{equation}
where 
\[F(\theta)=f(\theta)(-f(\theta)^4+f(\theta)^3g(\theta)-f(\theta)^2g(\theta)^2-f(\theta)g(\theta)^3+g(\theta)^4)^2,\] and 
\[G(\theta)=g(\theta)(f(\theta)^4-f(\theta)^3g(\theta)-f(\theta)^2g(\theta)^2+f(\theta)g(\theta)^3-g(\theta)^4)^2.\]
Note that $\theta^{3^m}=\theta$. Then \eqref{eq05} is reduced to
$$F(\theta)-\theta^{27}G(\theta)=0.$$
By Magma \cite{Magma}, the left-hand side of the above equation can be decomposed into the product of some irreducible factors as
\begin{align*}
&(\theta^{13}-\theta^{12}-\theta^{11}-\theta^{10}+\theta^9-\theta^7-\theta^5+\theta^4+\theta^3-1)\\
&(\theta^{13}-\theta^{12} -\theta^{11}-\theta^7-\theta^6-\theta^5-\theta^4-\theta^3+\theta^2-\theta-1)\\
&(\theta^{13}+\theta^{12}-\theta^{11}+\theta^{10}+\theta^9+\theta^8+\theta^7+\theta^6+\theta^2+\theta-1)\\
&(\theta^{13}+\theta^{12}-\theta^{10}+\theta^9-\theta^8+\theta^7+\theta^6+\theta^5+\theta^4-\theta^3+\theta^2-1)\\
&(\theta^{13}-\theta^{11}+\theta^{10}-\theta^9-\theta^8-\theta^7-\theta^6+\theta^5-\theta^4+\theta^3-\theta-1)\\
&(\theta^{13}-\theta^{10}-\theta^9+\theta^8+\theta^6-\theta^4+\theta^3+\theta^2+\theta-1)\\
&(\theta^9-\theta^7-\theta^6-\theta^5+\theta^4+\theta^2-1)(\theta^9-\theta^7-\theta^5+\theta^4+\theta^3+\theta^2-1)\\
&(\theta-1)^9(\theta+1)\theta=0
\end{align*}
over $\mathbb{F}_3$. Now from $\gcd(6,m)=1$ and Lemma \ref{lem_irr}, we can conclude that the last equation has no solution in $\mathbb{F}_{3^m}\setminus \mathbb{F}_3$ if $\gcd(m,13)=1$. This completes the proof.
\end{proof}

By Lemmas \ref{lem_main0} and \ref{lem_main1}, we obtain the following conclusion which is a partial answer to Problem \ref{prob7.9}.

\begin{theorem}\label{thm_main0}
Let $m>1$ be an odd integer with $m\equiv 7, 11\pmod {12}$ and $\gcd(m, 13)=1$. Let $e=\frac{3^{\frac{m+3}{2}}+5}{2}$. Then the ternary cyclic code $\mc{C}_{(1,e)}$ has parameters $[3^m-1,3^m-1-2m,4]$. 
\end{theorem}

Similarly to proving Theorem \ref{thm_main0}, we can prove the following Theorem \ref{thm_main01}.

\begin{theorem}\label{thm_main01}
Let $m>1$ be an odd integer with $m\equiv 1\pmod 6$ and $\gcd(m, 235)=1$. Let $e=\frac{3^{\frac{m+2}{3}}+5}{2}$. Then the ternary cyclic code $\mc{C}_{(1,e)}$ has parameters $[3^m-1,3^m-1-2m,4]$. 
\end{theorem}

\section{The third and fourth new families of optimal ternary cyclic codes}\label{sec_newcodes}

In this section, we present the third and fourth families of optimal ternary cyclic codes $\mc{C}_{(1,e)}$ with respect to the sphere packing bound. In both families, the exponents $e$ satisfy $e(3^h\pm 1)\equiv \frac{3^m-a}{2}\pmod{3^m-1}$ where $0\leq h\leq m-1$ and $a\equiv 3\pmod{4}$. Several recent work are special cases of our results.

\subsection{The third family of optimal ternary cyclic codes with minimum distance four}

In this subsection, we consider the exponents $e$ as the form of
\begin{equation}\label{eq_e1}
e(3^h-1)\equiv \frac{3^m-a}{2}\pmod{3^m-1},
\end{equation}
where $m,h$ and $a$ are integers with $1\leq h\leq m-1$ and $a\equiv 3\pmod{4}$. We first provide a necessary condition for the existence of such code $\mc{C}_{(1,e)}$ with parameters $[3^m-1,3^m-1-2m,4]$. Then we present a new family of ternary cyclic codes $\mc{C}_{(1,e)}$.

\begin{lemma}\label{lem_e1_exist}
Let $m, h$ and $a$ be integers with $1\leq h\leq m-1$ and $a\equiv 3\pmod{4}$. Then there exists an even integer $e$ satisfying 
\begin{equation*}
e(3^h-1)\equiv \frac{3^m-a}{2}\pmod{3^m-1}
\end{equation*}
if and only if $m$ is odd and $(3^{\gcd(h,m)}-1) \mid (a-1)$.
\end{lemma}
\begin{proof}
We first note that the congruence $e(3^h-1)\equiv \frac{3^m-a}{2}\pmod{3^m-1}$ has an even solution for $e$ if and only if $2k(3^h-1)\equiv \frac{3^m-a}{2}\pmod{3^m-1}$ has a solution for $k$. The latter condition holds if and only if 
\begin{equation}\label{eq_cr1_1}
\gcd(2(3^h-1),3^m-1)\left\vert \frac{3^m-a}{2}\right..
\end{equation}
If $m$ is even, then $3^m-a\equiv 2\pmod{4}$ which implies that $\frac{3^m-a}{2}$ is odd. But $\gcd(2(3^h-1),3^m-1)$ is even. Thus, $\gcd(2(3^h-1),3^m-1)\nmid \frac{3^m-a}{2}$, which implies that $e$ cannot exist. 

Next, we consider the case when $m$ is odd. We have $3^m-1\equiv 2\pmod{4}$ which means that $(3^m-1)_2=2$. Thus by Lemma \ref{lem_gcd}, we have 
\[\gcd(2(3^h-1),3^m-1)=\gcd(3^h-1,3^m-1)=3^{\gcd(h,m)}-1.\] 
So \eqref{eq_cr1_1} holds if and only if
\begin{equation}\label{eq_cr1_2}
\gcd\left(3^{\gcd(h,m)}-1,\frac{3^m-a}{2}\right)=3^{\gcd(h,m)}-1.
\end{equation}
Since $m$ is odd, then $\gcd(h,m)$ is odd which implies that $3^{\gcd(h,m)}-1\equiv 2\pmod{4}$. Note that $3^m-a\equiv 0\pmod{4}$. Then \eqref{eq_cr1_2} holds if and only if
\begin{equation}\label{eq_cr1_3}
\gcd(3^{\gcd(h,m)}-1,3^m-a)=3^{\gcd(h,m)}-1.
\end{equation}
Since $\gcd(h,m)\mid m$, then $(3^{\gcd(h,m)}-1)\mid (3^m-1)$. So we have 
\[\gcd(3^{\gcd(h,m)}-1,3^m-a)=\gcd(3^{\gcd(h,m)}-1,3^m-1+1-a)=\gcd(3^{\gcd(h,m)}-1,a-1).\] Therefore \eqref{eq_cr1_3} holds if and only if $\gcd(3^{\gcd(h,m)}-1,a-1)=3^{\gcd(h,m)}-1$ which is equivalent to $(3^{\gcd(h,m)}-1)\mid (a-1)$. This completes the proof.
\end{proof}

We will construct a family of optimal ternary cyclic codes $\mc{C}_{(1,e)}$, where $e$ satisfies \eqref{eq_e1} with 
\[a=2\cdot 3^t\delta+1,\] 
where $t$ is a nonnegative integer and $\delta\in\{1,-1\}$. 

It is clear that $a=2\cdot 3^t\delta+1\equiv 3\pmod{4}$. For these values of $a$, we are able to reformulate the conditions for the existence of an even integer $e$ in Lemma \ref{lem_e1_exist} in the following simpler way. This is done by noting that $a-1=2\cdot 3^t\delta$ and $\gcd(3,3^{\gcd(h,m)}-1)=1$.


\begin{lemma}\label{lem_e1_exist2}
Let $m$ and $h$ be integers with $1\leq h\leq m-1$. Let $\delta\in\{1,-1\}$ and let $a=2\cdot 3^t\delta+1$ where $t$ is a nonnegative integer. Then there exists an even integer $e$ satisfying \eqref{eq_e1} if and only if $m$ is odd and $\gcd(h,m)=1$.
\end{lemma}

Next, we verify that an even solution $e$ of \eqref{eq_e1} satisfies that $e\notin C_1$ and $|C_e|=m$.

\begin{lemma}\label{lem_econd1}
Let $m>1$ be an odd integer and let $h$ be an integer such that $1\leq h\leq m-1$ and $\gcd(h,m)=1$. Let $\delta\in\{1,-1\}$ and let $a=2\cdot 3^t\delta+1$ where $t$ is a nonnegative integer. Suppose that $e$ is an even integer satisfying \eqref{eq_e1}, then $e\notin C_1$ and $|C_e|=m$. 
\end{lemma}
\begin{proof}
Since $e$ is even, it is clear that $e\notin C_1$. Since $m$ is odd, then $3^m-1\equiv 2\pmod{4}$ and $3^m-a\equiv 0\pmod{4}$. Thus 
\begin{align*}
&\gcd(e(3^h-1),3^m-1)=\gcd\left(\frac{3^m-a}{2},3^m-1\right)\\
=&\gcd(3^m-a,3^m-1)=\gcd(a-1,3^m-1)=\gcd(2\cdot 3^t,3^m-1)\\
=&2.
\end{align*}
This implies that $\gcd(e,3^m-1)=2$. By Lemma \ref{lem_Cesize}, we have $|C_e|=m$. This completes the proof.
\end{proof}

The following is a technical lemma. 

\begin{lemma}\label{lem_gcdcond}
Let $m>1$ be an odd integer and let $h$ be an integer such that $1\leq h\leq m-1$ and $\gcd(h,m)=1$. Let $\delta\in\{1,-1\}$ and let $a=2\cdot 3^t\delta+1$ where $t$ is a nonnegative integer. Suppose that $e$ is an even integer satisfying \eqref{eq_e1}, then $\gcd(e-3^t\delta,3^m-1)=1$. 
\end{lemma}
\begin{proof}
By \eqref{eq_e1}, we deduce that
\[e-3^t\delta\equiv e3^h-\frac{3^m-1}{2}\pmod{3^m-1}.\]
Thus, it follows that 
\[\gcd(e-3^t\delta,3^m-1)=\gcd\left(e3^h-\frac{3^m-1}{2},3^m-1\right).\]
Since $\frac{3^m-1}{2}$ is odd, the above equation is reduced to 
\[\gcd(e-3^t\delta,3^m-1)=\gcd\left(e3^h-\frac{3^m-1}{2},\frac{3^m-1}{2}\right)=\gcd\left(e3^h,\frac{3^m-1}{2}\right)=\gcd\left(e,\frac{3^m-1}{2}\right).\]
Recall from the proof of Lemma \ref{lem_econd1} that $\gcd(e,3^m-1)=2$. We conclude that
\[\gcd(e-3^t\delta,3^m-1)=\gcd\left(e,\frac{3^m-1}{2}\right)=1.\]
This completes the proof.
\end{proof}

We are ready to present our first new family of optimal ternary cyclic codes.

\begin{theorem}\label{thm_main1}
Let $m>1$ be an odd integer, and let $h$ be an integer such that $1\leq h\leq m-1$ and $\gcd(h,m)=1$. Let $\delta\in\{1,-1\}$ and let $a=2\cdot 3^t\delta+1$ where $t$ is a nonnegative integer. Let $e$ be an even integer satisfying the following conditions: \vspace{0.3em}

{\rm 1)} $1<e<3^m-1$; \vspace{0.3em}

{\rm 2)} $e(3^h-1)\equiv \frac{3^m-a}{2}\pmod{3^m-1}$; \vspace{0.3em}

{\rm 3)} $\gcd(e-\frac{\delta+3}{2}\cdot 3^t,3^m-1)=\frac{\delta+3}{2}$. \vspace{0.3em}

Then the ternary cyclic code $\mc{C}_{(1,e)}$ has parameters $[3^m-1,3^m-1-2m,4]$.
\end{theorem}

\begin{proof}
By Lemmas \ref{lem_dingconds} and \ref{lem_econd1}, it suffices to show that the conditions C2 and C3 in Lemma \ref{lem_dingconds} hold. It is readily seen that $x=1$ is the only solution of $(x+1)^e+x^e+1=0$ in $\F_3$, and $x=0$ is the only solution of $(x+1)^e-x^e-1=0$ in $\F_3$. We shall prove that 
$$(x+1)^e=\pm(x^e+1)$$ 
has no solutions in $\F_{3^m}\setminus\F_3$. Suppose to the contrary that there exists $\theta\in\F_{3^m}\setminus\F_3$ such that $(\theta+1)^e=\pm(\theta^e+1)$. Raising both sides to the $(3^h-1)$-th power, we have
\[(\theta+1)^{e(3^h-1)}=(\theta^e+1)^{3^h-1}.\]
By \eqref{eq_e1}, the above equation reduces to 
\begin{equation}\label{eq_1}
(\theta+1)^{\frac{3^m-1}{2}-3^t\delta}(\theta^e+1)=\theta^{\frac{3^m-1}{2}-3^t\delta+e}+1.
\end{equation}
Let $\eta$ be the quadratic character of $\F_{3^m}$. Then \eqref{eq_1} translates to
\begin{equation}\label{eq_2}
\eta(\theta+1)(\theta+1)^{-3^t\delta}(\theta^e+1)=\eta(\theta)\theta^{-3^t\delta}\theta^e+1.
\end{equation}
We derive a contradiction from \eqref{eq_2} in the following four cases.

{\bf Case 1}. $\eta(\theta+1)=\eta(\theta)=1$. 

If $\delta=1$, then \eqref{eq_2} is reduced to $\theta^{3^t}+\theta^{-3^t+e}=0$, which is equivalent to $\theta^{e-2\cdot 3^t}=-1$. But this is contrary to the assumption that $\eta(\theta)=1$.

If $\delta=-1$, then \eqref{eq_2} is reduced to $\theta^{3^t}+\theta^e=0$, which is equivalent to $\theta^{e-3^t}=-1$. This is again contrary to the assumption that $\eta(\theta)=1$.

%

{\bf Case 2}. $\eta(\theta+1)=\eta(\theta)=-1$. 

If $\delta=1$, then \eqref{eq_2} is reduced to
\[\theta^{-3^t+e}+1=\theta^{3^t},\]
which is equivalent to $\theta^e+1=(\theta^{3^t}+1)^2$. Thus 
\[(\theta+1)^e=\pm(\theta^e+1)=\pm(\theta+1)^{2\cdot 3^t}.\]
Since $e$ is even and $\eta(-1)=-1$, we deduce that
\[(\theta+1)^{e-2\cdot 3^t}=1.\]
Since $\gcd(e-2\cdot 3^t,3^m-1)=2$, then $\theta+1=\pm 1$, which leads to a contradiction. 

If $\delta=-1$, then \eqref{eq_2} is reduced to 
\[\theta^e+1=-(\theta^{3^t}+1).\]
Thus
\[(\theta+1)^e=\pm(\theta^e+1)=\mp(\theta+1)^{3^t}.\]
Since $\eta(\theta+1)=-1$ and $\eta(-1)=-1$, we deduce that 
\[(\theta+1)^{e-3^t}=-1.\]
Since $\gcd(e-3^t,3^m-1)=1$, then $\theta+1=-1$, which leads to a contradiction. 

{\bf Case 3}. $\eta(\theta+1)=1$, and $\eta(\theta)=-1$. 

If $\delta=1$, then \eqref{eq_2} is reduced to 
\begin{equation}\label{eq_f1c3_1}
\theta^{2\cdot 3^t}+\theta^{e+3^t}-\theta^e=0,
\end{equation} 
which implies that 
\[\theta^e+1=-\frac{(\theta^{3^t}+1)^2}{\theta^{3^t}-1}.\]
Thus
\[(\theta+1)^e=\pm(\theta^e+1)=\mp \frac{(\theta^{3^t}+1)^2}{\theta^{3^t}-1}.\]
On the other hand, we deduce from \eqref{eq_f1c3_1} that $\theta^{3^t}-1=-\theta^{2\cdot 3^t-e}$. Noting that $\eta(\theta+1)=-\eta(\theta)$ and $\eta(-1)=-1$, we have
\[\left(\frac{\theta+1}{\theta}\right)^{e-2\cdot 3^t}=1.\]
Since $\gcd(e-2\cdot 3^t,3^m-1)=2$, the above equation implies that $\theta+1=\pm \theta$, which leads to a contradiction.

If $\delta=-1$, then \eqref{eq_2} is reduced to 
\begin{equation}\label{eq_case3_1}
\theta^{3^t+e}-\theta^{3^t}-\theta^e=0.
\end{equation}
On the one hand, we can deduce that 
\[(\theta^{3^t}-1)(\theta^e+1)=-\theta^{3^t}-1.\]
Thus 
\begin{equation}\label{eq_case3_2}
\theta^e+1=-\frac{(\theta+1)^{3^t}}{\theta^{3^t}-1}.
\end{equation}
On the other hand, we deduce from \eqref{eq_case3_1} that $\theta^e-\theta^{e-3^t}=1$. So 
\begin{equation}\label{eq_case3_3}
\theta^{e-3^t}(\theta^{3^t}-1)=1.
\end{equation}
By \eqref{eq_case3_2} and \eqref{eq_case3_3}, we have
\[(\theta+1)^e=\pm (\theta^e+1)=\mp (\theta+1)^{3^t}\theta^{e-3^t}.\]
Noting that $\eta(\theta+1)=-\eta(\theta)$ and $\eta(-1)=-1$, we have
\[\left(\frac{\theta+1}{\theta}\right)^{e-3^t}=-1.\]
Since $\gcd(e-3^t,3^m-1)=1$, we deduce that $\theta+1=-\theta$, which leads to a contradiction. 

{\bf Case 4}. $\eta(\theta+1)=-1$ and $\eta(\theta)=1$. 

If $\delta=1$, then \eqref{eq_2} is reduced to $(\theta^{3^t}-1)(\theta^{e-3^t}-1)=0$. Since $\theta\neq 1$, then $\theta^{e-3^t}-1=0$. By Lemma \ref{lem_gcdcond}, we have $\theta=1$, which leads to a contradiction.

If $\delta=-1$, then \eqref{eq_2} is reduced to 
\[(\theta^{3^t}-1)(\theta^e-1)=0.\]
Since $\theta\neq 1$, we have $\theta^e=1$. Recall from the proof of Lemma \ref{lem_econd1} that $\gcd(e, 3^m-1)=2$. We conclude that $\theta=\pm 1$, which leads to a contradiction.  

In the four cases discussed above, we have demonstrated that $(x+1)^e=\pm(x^e+1)$ has no solutions in $\F_{3^m}\setminus\F_3$, and therefore conditions C2 and C3 of Lemma \ref{lem_dingconds} are satisfied. The proof is now complete.
\end{proof}

\begin{remark}
The condition $\gcd(e-\frac{\delta+3}{2}\cdot 3^t,3^m-1)=\frac{\delta+3}{2}$ in Theorem \ref{thm_main1} cannot be removed. See the following counterexamples. 

When $\delta=1$ and $t=0$, we have $a=3$. For $m=11, h=4$, we take $e=129538$. Then $\gcd(e-2,3^m-1)=46$. We check by Magma \cite{Magma} that the equation $(x+1)^e-x^e-1=0$ has 23 solutions in $\F_{3^m}$. So the code $\mc{C}_{(1,e)}$ does not have parameters $[3^m-1,3^m-1-2m,4]$. 

When $\delta=-1$ and $t=0$, we have $a=-1$. For $m=11, h=7$, we take $e=40964$. Then $\gcd(e-1,3^m-1)=23$. We check by Magma \cite{Magma} that the equation $(x+1)^e-x^e-1=0$ has 23 solutions in $\F_{3^m}$. So the code $\mc{C}_{(1,e)}$ does not have parameters $[3^m-1,3^m-1-2m,4]$.
\end{remark}

\begin{remark}
In \cite{YH19}, the authors showed that $C_{(1,e)}$ has parameters $[3^m-1,3^m-1-2m,4]$ if $m$ is odd and $e=\frac{3^m-3}{4}$. By setting $h=1, t=0, \delta=1$ (so $a=3$), we see that their result is a special case of Theorem \ref{thm_main1}.
\end{remark}

\begin{example}
Let $m =5$ and $(h, t, \delta)=(1, 1, -1)$. Then $a=-5$ and \eqref{eq_e1} has only one even solution $e=62$. It can be checked that $\gcd(e-\frac{\delta+3}{2}\cdot 3^t,3^m-1)=\frac{\delta+3}{2}$. Let $\alpha$ be a primitive element of $\F_{3^m}$ with $\alpha^5+2\alpha+1=0$. According to Theorem \ref{thm_main1}, $\mc{C}_{(1,e)}$ is a ternary cyclic code with parameters $[242, 232, 4]$. The generator polynomial of $\mc{C}_{(1,e)}$ is $x^{10}+2x^9+2x^8+x^7+2x^6+x^5+x^3+x^2+x+2$.
\end{example}

\begin{example}
Let $m =7$ and $(h, t, \delta)=(1, 1, 1)$. Then $a=7$ and \eqref{eq_e1} has only one even solution $e=1638$. It can be checked that $\gcd(e-\frac{\delta+3}{2}\cdot 3^t,3^m-1)=\frac{\delta+3}{2}$. Let $\alpha$ be a primitive element of $\F_{3^m}$ with $\alpha^7+2\alpha^2+1=0$. According to Theorem \ref{thm_main1}, $\mc{C}_{(1,e)}$ is a ternary cyclic code with parameters $[2186, 2172, 4]$. The generator polynomial of $\mc{C}_{(1,e)}$ is $x^{14}+2x^{13}+2x^{12}+2x^{10}+2x^8+x^7+2x^6+2x^4+x^3+2x^2+x+2$.
\end{example}

\subsection{The fourth family of optimal ternary cyclic codes with minimum distance four}

In this subsection, we consider the exponents $e$ as the form of
\begin{equation}\label{eq_e2}
e(3^h+1)\equiv \frac{3^m-a}{2}\pmod{3^m-1},
\end{equation}
where $1\leq h\leq m-1$ and $a\equiv 3\pmod{4}$. We give a necessary condition for the existence of such code $\mc{C}_{(1,e)}$ with parameters $[3^m-1,3^m-1-2m,4]$. Then we present a new family of such codes.

\begin{lemma}
Let $m,h$ and $a$ be integers with $1\leq h\leq m-1$ and $a\equiv 3\pmod{4}$. Then there exists an even integer $e$ satisfying
\begin{equation*}
e(3^h+1)\equiv\frac{3^m-a}{2}\pmod{3^m-1}
\end{equation*} 
 if and only if $m$ is odd.
\end{lemma}
\begin{proof}
As we saw in the proof of Lemma \ref{lem_e1_exist}, an even integer $e$ satisfying \eqref{eq_e2} exists if and only if 
\[\gcd(2(3^h+1),3^m-1)\left\vert \frac{3^m-a}{2}\right..\]
If $m$ is even, then $3^m-a\equiv 2\pmod{4}$ which implies that $\frac{3^m-a}{2}$ is odd. However, it is clear that $\gcd(2(3^h+1),3^m-1)$ is even. Thus, $\gcd(2(3^h+1),3^m-1)\nmid\frac{3^m-a}{2}$, which establishes the nonexistence of $e$. 

Next, we assume that $m$ is odd. Then $3^m-a\equiv 0\pmod{4}$ which implies that $\left(\frac{3^m-a}{2}\right)_2\geq 2$. On the other hand, by Lemma \ref{lem_gcd}, we have $\gcd(3^h+1,3^m-1)=2$. Since $3^m-1\equiv 2\pmod{4}$, then $\gcd(2(3^h+1),3^m-1)=2\mid \frac{3^m-a}{2}$. This completes the proof.
\end{proof}

We will construct a family of optimal ternary codes $\mc{C}_{(1,e)}$, where $e$ satisfies \eqref{eq_e2} with 
\[a=-2\cdot 3^t+1,\] 
where $t$ is a nonnegative integer. It is clear that $a=-2\cdot 3^t+1\equiv 3\pmod{4}$.

The next two lemmas can be proved in the same way as Lemmas \ref{lem_econd1} and \ref{lem_gcdcond}. Therefore, we omit the proofs.

\begin{lemma}\label{lem_econd2}
Let $m>0$ be an odd integer and let $h$ be an integer with $0\leq h\leq m-1$. Let $a=-2\cdot 3^t+1$ where $t\geq 0$ is an integer. Suppose that $e$ is an even integer satisfying \eqref{eq_e2}. Then $e\notin C_1$ and $|C_e|=m$. 
\end{lemma}
%
\begin{lemma}\label{lem_f2_gcd}
Let $m>0$ be an odd integer and let $h$ be an integer with $0\leq h\leq m-1$. Let $a=-2\cdot 3^t+1$ where $t\geq 0$ is an integer. Suppose that $e$ is an even integer satisfying \eqref{eq_e2}. Then $\gcd(e-3^t,3^m-1)=1$. 
\end{lemma}

We are ready to present our second new family of optimal ternary cyclic codes.

\begin{theorem}\label{thm_main2}
Let $m>0$ be an odd integer and let $h$ be an integer with $0\leq h\leq m-1$. Let $a=-2\cdot 3^t+1$ where $t\geq 0$ is an integer. Let $e$ be an even integer satisfying \eqref{eq_e2} and $1<e<3^m-1$. Then the ternary cyclic code $\mc{C}_{(1,e)}$ has parameters $[3^m-1,3^m-1-2m,4]$. 
\end{theorem}
\begin{proof}
By Lemmas \ref{lem_dingconds} and \ref{lem_econd2}, we only need to show that the conditions C2 and C3 in Lemma \ref{lem_dingconds} hold. As we have explained in the proof of Theorem \ref{thm_main1}, we only need to show that
\begin{equation}\label{eq_econd2}
(x+1)^e=\pm (x^e+1)
\end{equation}
has no solutions in $\F_{3^m}\setminus\F_3$. Suppose to the contrary that $\theta\in\F_{3^m}\setminus\F_3$ is a solution to \eqref{eq_econd2}. Then 
\[(\theta+1)^{e(3^h+1)}=(\theta^e+1)^{3^h+1},\]
which can be written as
\begin{equation}\label{eq_f2cond}
\eta(\theta+1)(\theta^{3^t}+1)=\eta(\theta)\theta^{3^t}+\eta(\theta)\theta^{3^t-e}+\theta^e+1,
\end{equation}
where $\eta$ is the quadratic character of $\F_{3^m}$. We derive a contradiction in the following four cases.

{\bf Case 1}. $\eta(\theta+1)=\eta(\theta)=1$. 

Then \eqref{eq_f2cond} is reduced to $\theta^{3^t-e}+\theta^e=0$ which is equivalent to $\theta^{2e-3^t}=-1$. This cannot happen because $-1$ is a nonsquare in $\F_{3^m}$ but $\theta$ is a square.

{\bf Case 2}. $\eta(\theta+1)=\eta(\theta)=-1$.

Then \eqref{eq_f2cond} is reduced to $\theta^{3^t-e}-\theta^e+1=0$, which can be written as 
\[y^2-y-\theta^{3^t}=0,\]
where $y=\theta^e$. Note that $\Delta_y=1+\theta^{3^t}=(1+\theta)^{3^t}$ is a nonsquare. Then there are no solutions for $y$.

{\bf Case 3}. $\eta(\theta+1)=1$ and $\eta(\theta)=-1$. 

Then \eqref{eq_f2cond} is reduced to $\theta^{3^t}-\theta^{3^t-e}+\theta^e=0$ which can be written as 
\[y^2+\theta^{3^t}y-\theta^{3^t}=0,\]
where $y=\theta^e$. Note that $\Delta_y=\theta^{2\cdot 3^t}+\theta^{3^t}=\theta^{3^t}(\theta+1)^{3^t}$ which is a nonsquare. Then there are no solutions for $y$.

{\bf Case 4}. $\eta(\theta+1)=-1$ and $\eta(\theta)=1$. 

Then \eqref{eq_f2cond} is reduced to $\theta^{3^t}-\theta^{3^t-e}-\theta^e+1=(\theta^e-1)(\theta^{3^t-e}-1)=0$. If $\theta^e=1$, then $\theta=\pm 1$ since $\gcd(e,3^m-1)=2$. If $\theta^{3^t-e}=1$, then $\theta=1$ since $\gcd(e-3^t,3^m-1)=1$ by Lemma \ref{lem_f2_gcd}. Both are contradictions to the assumption that $\theta\notin\F_3$. 

In the four cases discussed above, we have demonstrated that $(x+1)^e=\pm(x^e+1)$ has no solutions in $\F_{3^m}\setminus\F_3$, and therefore conditions C2 and C3 of Lemma \ref{lem_dingconds} are satisfied. This completes the proof.
\end{proof}

\begin{remark}
Theorem 3.3 of \cite{ZH20} is a special case of Theorem \ref{thm_main2}, i.e., when $t=0$. We should mention that the proof idea of Theorem \ref{thm_main2} follows from \cite{ZH20}. 
\end{remark}

\begin{remark}
In \cite{DH13}, the authors showed that $\mc{C}_{(1,e)}$ has parameters $[3^m-1,3^m-1-2m,4]$ if $m\equiv 1\pmod{4}$ and $e\in\left\{\frac{3^{(m+1)/2}-1}{2}+\frac{3^m-1}{2}, \frac{3^{m+1}-1}{8}+\frac{3^m-1}{2}\right\}$, or $m\equiv 3\pmod{4}$ and $e\in\left\{\frac{3^{(m+1)/2}-1}{2}, \frac{3^{m+1}-1}{8}\right\}$, or $m\geq 3$ is odd and $e=\frac{3^m+1}{4}+\frac{3^m-1}{2}$. All these results are special cases of Theorem \ref{thm_main2}. This can be seen by noting that $\frac{3^{(m+1)/2}-1}{2}(3^{(m+1)/2}+1)\equiv \frac{3^m-1}{8}(3^1+1)\equiv \frac{3^m+1}{4}(3^0+1)\equiv \frac{3^m+1}{2}\pmod{3^m-1}$.
\end{remark}

\begin{example}
Let $m =7$ and $(h, t)=(1, 1)$. Then $a=-5$ and \eqref{eq_e1} has only one even solution $e=274$. Let $\alpha$ be a primitive element of $\F_{3^m}$ with $\alpha^7+2\alpha^2+1=0$. According to Theorem \ref{thm_main2}, $\mc{C}_{(1,e)}$ is a ternary cyclic code with parameters $[2186, 2172, 4]$. The generator polynomial of $\mc{C}_{(1,e)}$ is  $x^{14}+x^{13}+x^{10}+2x^9+x^8+x^6+2x^5+2x^3+x^2+2x+2$. 
\end{example}

\section{Further results on cyclic codes $\mc{C}_{(1,e)}$ with $e(3^h\pm 1)\equiv\frac{3^m-a}{2}\pmod{3^m-1}$}\label{sec_nonexist}

In Section \ref{sec_newcodes}, we construct new families of optimal cyclic codes $\mc{C}_{(1,e)}$ where $e$ satisfies $e(3^h\pm 1)\equiv\frac{3^m-a}{2}\pmod{3^m-1}$ for some integers $a$ with $a\equiv 3\pmod{4}$. In this section, we consider the cyclic codes $\mc{C}_{(1,e)}$ where $e$ satisfies $e(3^h\pm 1)\equiv\frac{3^m-a}{2}\pmod{3^m-1}$ for some integers $a$ with $a\equiv 1\pmod{4}$. Some nonexistence results on cyclic codes $\mc{C}_{(1,e)}$ with parameters $[3^m-1,3^m-1-2m,4]$ are presented.

\subsection{Cyclic codes $\mc{C}_{(1,e)}$ with $e(3^h\pm 1)\equiv\frac{3^m-a}{2}\pmod{3^m-1}$ where $a\equiv 5\pmod{8}$}

In this subsection, we show that a cyclic code $\mc{C}_{(1,e)}$ with $e(3^h\pm 1)\equiv\frac{3^m-a}{2}\pmod{3^m-1}$ where $a\equiv 5\pmod{8}$, cannot have parameters $[3^m-1,3^m-1-2m,4]$.

\begin{theorem}\label{thm_pma5}
Let $m, h$ and $e$ be integers with $0\leq h\leq m-1$ and $1<e<3^m-1$. Let $a$ be an integer with $a\equiv 5\pmod{8}$. Suppose that $e\notin C_1$ and 
\begin{equation}\label{eq_pma}
e(3^h\pm1)\equiv \frac{3^m-a}{2}\pmod{3^m-1}, 
\end{equation}
then the ternary cyclic code $\mc{C}_{(1,e)}$ does not have parameters $[3^m-1,3^m-1-2m,4]$.
\end{theorem}
\begin{proof}
It is enough to show that there are no even integers $e$ satisfying \eqref{eq_pma}. We have established in the proof of Lemma \ref{lem_e1_exist} that an even integer $e$ satisfying \eqref{eq_pma} exists if and only if 
\[\gcd(2(3^h\pm1),3^m-1)\left\vert \frac{3^m-a}{2}\right..\]
If $m$ is odd, then $3^m-a\equiv 2\pmod{4}$, which implies that $\frac{3^m-a}{2}$ is odd. However,  $\gcd(2(3^h\pm1),3^m-1)$ is even. Hence, $\gcd(2(3^h\pm 1),3^m-1)$ does not divide $\frac{3^m-a}{2}$. If $m$ is even, then $3^m-a\equiv 4\pmod{8}$, which implies that $(3^m-a)_2=4$. Note that $3^m-1\equiv 0\pmod{8}$ and $3^h\pm 1$ is even. So $\gcd(2(3^h\pm 1),3^m-1)_2\geq 4$. Therefore, $\gcd(2(3^h\pm1),3^m-1)$ does not divide $\frac{3^m-a}{2}$.  This completes the proof. 
\end{proof}

\begin{remark}\label{rmk_eeven}
The first part of the above proof is valid for all integers $a$ with $a\equiv 1\pmod{4}$. That is, there are no even integers $e$ satisfying $e(3^h\pm1)\equiv \frac{3^m-a}{2}\pmod{3^m-1}$ with $a\equiv 1\pmod{4}$ if $m$ is odd.
\end{remark}

\subsection{Cyclic codes $\mc{C}_{(1,e)}$ with $e(3^h+1)\equiv\frac{3^m-a}{2}\pmod{3^m-1}$ where $a\equiv 1\pmod{8}$}

The situation when $a\equiv 1\pmod{8}$ is much more complicated than when $a\not\equiv 1\pmod{8}$. We focus on the case when $a=1$ and show that the codes $\mc{C}_{(1,e)}$ cannot have parameters $[3^m-1,3^m-1-2m,4]$. For the ease of notation, we define $(0)_2=\infty$ in this subsection.

\begin{lemma}\label{lem_p1m1e}
Let $m$ and $h$ be integers with $0\leq h\leq m-1$. Then there exists an even integer $e$ such that 
\begin{equation}\label{eq_b}
e(3^h+1)\equiv \frac{3^m-1}{2}\pmod{3^m-1}
\end{equation}
if and only if $m\equiv 0\pmod{4}$, or $m\equiv 2\pmod{4}$ and $(h)_2\geq (m)_2$.
\end{lemma}
\begin{proof}
We have seen in the proof of Lemma \ref{lem_e1_exist} that an even integer $e$ satisfying \eqref{eq_b} exists if and only if 
\[\gcd(2(3^h+1),3^m-1)\left\vert \frac{3^m-1}{2}\right..\]
%
By Remark \ref{rmk_eeven}, we only need to consider the case when $m$ is even. 
If $m\equiv 0\pmod{4}$, then $3^m-1\equiv 0\pmod{16}$ which implies that $\left(\frac{3^m-1}{2}\right)_2\geq 8$. Let $d=\gcd(2(3^h+1),3^m-1)$. It is clear that $(d)_{2'}\mid \left(\frac{3^m-1}{2}\right)_{2'}$. Thus in order to show that $d\mid \frac{3^m-1}{2}$, it suffices to prove that $(d)_2\leq 8$. Note that 
\[
\begin{array}{l}
\vspace{0.2cm} 3^h+1\equiv
\left\{\begin{array}{ll}
\vspace{0.1cm} 4\pmod{8}, &\text{if }h\ \text{is odd },\\
2\pmod{4},&\text{otherwise}.
\end{array}
\right .
\end{array}
\]
So $(3^h+1)_2\leq 4$ and therefore $(d)_2\leq 8$. 

Finally, we consider the case when $m\equiv 2\pmod{4}$. In this case, we have $3^m-1\equiv 8\pmod{16}$. If $(h)_2\geq (m)_2$, by Lemma \ref{lem_gcd}, we have $\gcd(3^h+1,3^m-1)=2$. Thus $\gcd(2(3^h+1),3^m-1)=4$ which divides $\frac{3^m-1}{2}$. If $(h)_2<(m)_2$, then $h$ is odd since $m\equiv 2\pmod{4}$. Thus, $3^h+1\equiv 0\pmod{4}$, which implies that 
$$(\gcd(2(3^h+1),3^m-1))_2\geq 8.$$
However, since $\left(\frac{3^m-1}{2}\right)_2=4$, we see that $\gcd(2(3^h+1),3^m-1)\nmid \frac{3^m-1}{2}$. This completes the proof.
\end{proof}

\begin{lemma}\label{lem_s5_2}
Let $m$ be a positive integer with $m\equiv 0 \pmod{4}$, and let $h$ be an integer with $0\leq h\leq m-1$ and $h\neq \frac{m}{2}$. Suppose that $e$ is an even integer satisfying \eqref{eq_b}. Then $|C_e|\neq m$.
\end{lemma}
\begin{proof}
By \eqref{eq_b}, we have $\gcd(e(3^h+1),3^m-1)=\gcd(\frac{3^m-1}{2},3^m-1)=\frac{3^m-1}{2}$. By Lemma \ref{lem_gcd}, we have
\[\gcd(3^h+1,3^m-1)=\begin{cases}
3^{\gcd(h,m)}+1,&\text{if }(h)_2<(m)_2,\\
2,&\text{if }(h)_2\geq (m)_2.
\end{cases}\]
Then it follows that
\[\begin{cases}
\vspace{0.2cm}\frac{3^m-1}{2(3^{\gcd(h,m)}+1)}\mid e,&\text{if }(h)_2<(m)_2,\\
\frac{3^m-1}{4} \mid e,&\text{if }(h)_2\geq (m)_2.\end{cases}\]

If $(h)_2\geq (m)_2$, then $(3^m-1)\mid 4e$ which implies that $9e\equiv e\pmod{3^m-1}$. Thus $|C_e|\leq 2$. 

If $(h)_2<(m)_2$, then $2\cdot\gcd(h,m)<m$ since $h\neq \frac{m}{2}$. Note that \[(3^m-1)\mid 2e(3^{\gcd(h,m)}+1)\mid e(3^{2\gcd(h,m)}-1).\] 
We conclude that $|C_e|\leq 2\cdot\gcd(h,m)<m$. This completes the proof.
\end{proof}

\begin{lemma}\label{lem_s5_3}
Let $m$ be a positive integer with $m\equiv 0 \pmod{4}$, and let $h=\frac{m}{2}$. Suppose that $e$ is an even integer satisfying \eqref{eq_b}. Then $(x+1)^e-x^e-1=0$ has solutions in $\F_{3^m}\setminus\F_3$.
\end{lemma}
\begin{proof}
Since $h=\frac{m}{2}$, we have $e=\frac{3^h-1}{2}+(3^h-1)k$ for some integer $k$. Viewing $(x+1)^e-x^e-1=0$ as an equation over $\F_{3^h}$, it can be reduced to $(x+1)^\frac{3^h-1}{2}-x^{\frac{3^h-1}{2}}-1=0$. Recall from Section \ref{sec_prem} that $x^{\frac{3^h-1}{2}}=1$ or $-1$ according to whether $x$ is a square of $\F_{3^h}^*$ or not. We shall show that there exists $\theta\in\F_{3^h}\setminus\F_3$ such that $\theta$ is a square in $\F_{3^h}$ and $\theta+1$ is a nonsquare in $\F_{3^h}$.

Consider the equation $uv=1$ over $\F_{3^h}$, which clearly has $3^h-1$ solutions for $(u,v)$. It is easy to see that $(u,v)\mapsto (u+v,u-v)$ is a bijection from $\F_{3^h}^*\times \F_{3^h}^*$ to itself. Thus the number of solutions $(u,v)$ for $u^2-v^2=1$ over $\F_{3^h}$ is also $3^h-1$. It follows that there exists $s\in\F_{3^h}$ such that $s^2+1=u^2$ has no solutions for $u$. Let $\theta=s^2$. Then $\theta+1$ is a nonsquare of $\F_{3^h}$. 
Therefore $(\theta+1)^{\frac{3^h-1}{2}}-\theta^{\frac{3^h-1}{2}}-1=-1-1-1=0$. 

Finally, we show that $\theta\notin\F_3$. If $\theta=0$, then $\theta+1=1=u^2$ has solutions in $\F_{3^h}$. If $\theta=1$, then $\theta+1=-1=u^2$ has solutions in $\F_{3^h}$ since $-1$ is a square when $h$ is even. If $\theta=-1$, then $\theta+1=0=u^2$ has solutions in $\F_{3^h}$. Therefore $\theta\notin\F_3$. The proof is complete.
\end{proof}

\begin{lemma}\label{lem_s5_4}
Let $m$ be a positive integer with $m\equiv 2\pmod{4}$, and let $h$ be an integer with $0\leq h\leq m-1$ and $(h)_2\geq (m)_2$. Let $e$ be an even integer satisfying \eqref{eq_b}. Then $|C_e|=2$.
\end{lemma}
\begin{proof}
On the one hand, we demonstrate that  $|C_e|\leq 2$. Since $e$ satisfies \eqref{eq_b}, we have
\[\gcd(e(3^h+1),3^m-1)=\gcd\left(\frac{3^m-1}{2},3^m-1\right)=\frac{3^m-1}{2}.\]
Thus, it follows that 
\[\gcd(2e(3^h+1),3^m-1)=3^m-1.\]
From the proof of Lemma \ref{lem_p1m1e}, we know that $\gcd(2(3^h+1),3^m-1)=4$. Combining these results, we obtain the equation 
\[\gcd(2e(3^h+1),3^m-1)=4\cdot\gcd\left(e,\frac{3^m-1}{4}\right).\]
This implies that $\frac{3^m-1}{4}\mid e$. Consequently, we have $e\cdot 3^2\equiv e\pmod{3^m-1}$ and so $|C_e|\leq 2$. 

On the other hand, we show that  $|C_e|>1$. By Lemma \ref{lem_gcd}, we have $\gcd(3^h+1,3^m-1)=2$. Then it follows that 
\[\frac{3^m-1}{2}=\gcd(e(3^h+1),3^m-1)=2\cdot\gcd\left(e,\frac{3^m-1}{2}\right).\]
If $2e\equiv 0\pmod{3^m-1}$, then we would have $2\cdot\gcd(e,\frac{3^m-1}{2})=3^m-1$, which leads to a contradiction. We can thus conclude that $|C_e|=2$. This completes the proof.
\end{proof}

By combining Lemmas \ref{lem_p1m1e}, \ref{lem_s5_2}, \ref{lem_s5_3}, and \ref{lem_s5_4}, we obatin the following conclusion.
\begin{theorem}\label{thm_p1m1}
Let $m, h$ and $e$ be integers with $0\leq h\leq m-1$ and $1<e<3^m-1$. Suppose that $e(3^h+1)\equiv \frac{3^m-1}{2}\pmod{3^m-1}$ and $e\notin C_1$. Then the ternary cyclic code $\mc{C}_{(1,e)}$ does not have parameters $[3^m-1,3^m-1-2m,4]$.
\end{theorem}

\section{Concluding remarks}\label{sec_conclude}
\subsection{Summary of results of this paper}
In this paper, we consider the ternary cyclic codes $\mc{C}_{(1,e)}$ and determine their optimality with respect to the sphere packing bound in different cases. In the first part of this paper, we focus on two open problems of Ding and Helleseth proposed in \cite{DH13}. We present the first counterexamples for both problems and also construct new families of optimal codes which provide positive answers to Problem \ref{prob7.9}. In the second part, we study cyclic codes $\mc{C}_{(1,e)}$ with $e(3^h\pm 1)\equiv \frac{3^m-a}{2}\pmod{3^m-1}$ and $a$ odd. Two new families of optimal codes $\mc{C}_{(1,e)}$ with $e$ in this form are constructed. Several nonexistence results of such optimal codes are also obtained. The results in this work are summarized in Table \ref{tab_newDHfamily} and Table \ref{tab_summary}.

\begin{table}[h]\aboverulesep=0pt \belowrulesep=0pt
\setlength{\abovecaptionskip}{0cm}
\setlength{\belowcaptionskip}{0cm}
\caption{\footnotesize New optimal ternary cyclic codes $\mc{C}_{(1,e)}$ related to Ding and Helleseth's problems}
\label{tab_newDHfamily}
\centering
\[\footnotesize
\arraycolsep=5pt
\begin{tabular}{lll}
\toprule
$e$ ($e$ is even) & \text{Conditions} &  \text{Reference} \\
\midrule
&&\\[-0.4cm]
$\frac{3^{\frac{m+3}{2}}+5}{2}$\hspace{5em}  & $m>1$ odd, $m\equiv 7, 11\pmod {12}$ and $\gcd(m, 13)=1$ \hspace{3em} & (Theorem \ref{thm_main0})\\[0.5em]
$\frac{3^{\frac{m+2}{3}}+5}{2}$ & $m>1$ odd, $m\equiv 1\pmod 6$ and $\gcd(m, 235)=1$ & (Theorem \ref{thm_main01})\\[0.2em]
\bottomrule
\end{tabular}
\]
\end{table}

\begin{table}[H]\aboverulesep=0pt \belowrulesep=0pt
\setlength{\abovecaptionskip}{0cm}
\setlength{\belowcaptionskip}{0cm}
\caption{\footnotesize Summary of results on $\mc{C}_{(1,e)}$ with $e(3^h\pm 1)\equiv \frac{3^m-a}{2}\pmod{3^m-1}$, $a$ odd}
\label{tab_summary}
\centering
\[\footnotesize
\begin{tabularx}{\textwidth}{lX}
\toprule
$e$ ($e$ is even) &   \text{Result} \\
\midrule
&\\[-1em]
$e(3^h-1)\equiv \frac{3^m-a}{2}\pmod{3^m-1}$, $a\equiv 3\pmod{4}$ & Let $a=2\cdot 3^t\delta+1$ with $t\geq 0$ and $\delta\in\{1,-1\}$. If $m$ is odd, $\gcd(h,m)=1$ and $\gcd(e-\frac{\delta+3}{2}3^t,3^m-1)=\frac{\delta+3}{2}$, then $\mc{C}_{(1,e)}$ has parameters $[3^m-1,3^m-1-2m,4]$. \\
&\hfill(Theorem \ref{thm_main1})\\[0.1em]
$e(3^h+1)\equiv \frac{3^m-a}{2}\pmod{3^m-1}$, $a\equiv 3\pmod{4}$ & If $m$ is odd and $a=-2\cdot 3^t+1$ with $t\geq 0$, then $\mc{C}_{(1,e)}$ has parameters $[3^m-1,3^m-1-2m,4]$.\\
&\hfill (Theorem \ref{thm_main2})\\[0.1em]
$e(3^h-1)\equiv \frac{3^m-a}{2}\pmod{3^m-1}$, $a\equiv 1\pmod{4}$ &If $a\equiv 5\pmod{8}$, then $\mc{C}_{(1,e)}$ does not have parameters $[3^m-1,3^m-1-2m,4]$.\\
&\hfill (Theorem \ref{thm_pma5})\\[0.1em]
$e(3^h+1)\equiv \frac{3^m-a}{2}\pmod{3^m-1}$, $a\equiv 1\pmod{4}$ &If $a\equiv 5\pmod{8}$, then $\mc{C}_{(1,e)}$ does not have parameters $[3^m-1,3^m-1-2m,4]$.\\
&\hfill (Theorem \ref{thm_pma5})\\[0.1em]
& If $a=1$, then $\mc{C}_{(1,e)}$ does not have parameters $[3^m-1,3^m-1-2m,4]$.\\
&\hfill(Theorem \ref{thm_p1m1}) \\[0.1em]
\bottomrule
\end{tabularx}
\]
\end{table}

\subsection{Cyclic codes $\mc{C}_{(1,e)}$ with $e(3^h-1)\equiv \frac{3^m-a}{2}\pmod{3^m-1}$ where $a\equiv 1\pmod{8}$}

The only case in which we do not have a result for the codes $\mc{C}_{(1,e)}$ with $e(3^h-1)\equiv \frac{3^m-a}{2}\pmod{3^m-1}$ is when $a\equiv 1\pmod{8}$. We make further discussions for this case in this subsection.

We have explained in Remark \ref{rmk_nescond} that in order to construct cyclic codes $\mc{C}_{(1,e)}$ with parameters $[3^m-1,3^m-1-2m,4]$, one must choose $e$ to be even. Thus, the first step in constructing such codes is to determine the conditions under which there exists an even integer $e$ that satisfies the prescribed property. We do not intend to study this problem here for all $a$ with $a\equiv 1\pmod{8}$, but only touch on the case $a=1$ as a starting point for future research. For $a=1$, we have the following result.

\begin{proposition}\label{lem_mm_econd}
Let $m>1, h$ be positive integers with $1\leq h\leq m-1$. Then there exists an even integer $e$ such that 
\begin{equation}\label{eq_e4}
e(3^h-1)\equiv \frac{3^m-1}{2}\pmod{3^m-1}
\end{equation}
if and only if $m\equiv 0\pmod{8}$ and $h\not \equiv0\pmod{\frac{(m)_2}{2}}$, or $m\equiv 2, 4, 6\pmod{8}$ and $h$ is odd.
\end{proposition}
\begin{proof}
An even integer $e$ satisfying \eqref{eq_e4} exists if and only if 
\[\gcd(2(3^h-1),3^m-1)\left\vert \frac{3^m-1}{2}\right..\]
The latter condition holds if and only if
\begin{equation}\label{dd}
4\cdot(3^h-1)_2\leq (3^m-1)_2.
\end{equation}
We divide the proof into the following three cases:

{\bf Case 1}. $m\equiv 1 \pmod{2}$. 

In this case, we have $3^m-1\equiv 2\pmod{4}$. Then, \eqref{dd} cannot happen. 

{\bf Case 2}. $m\equiv 0\pmod{8}$. 

In this case, let $M=\frac{(m)_2}{2}$. If $h\equiv 0\pmod{M}$, then there exists a positive integer $t$ such that $h=Mt$. Consequently, we have 
\[(3^h-1)_2=(3^{Mt}-1)_2=(3^{M}-1)_2(3^{M(t-1)}+\cdots+3^{M}+1)_2\geq (3^{M}-1)_2.\] 
On the other hand, we have
\begin{align*}
&(3^m-1)_2=(3^{(m)_2(m)_{2'}}-1)_2\\
=&(3^{(m)_2}-1)_2(3^{(m)_2((m)_{2'}-1)}+\cdots+3^{(m)_2}+1)_2\\
=&(3^{(m)_2}-1)_2=(3^{2M}-1)_2=(3^M+1)_2(3^M-1)_2\\
=&2\cdot(3^M-1)_2.
\end{align*}
Thus \eqref{dd} cannot happen. 

If $h\not\equiv 0\pmod{M}$ and $h$ is odd, then $3^h-1\equiv 2\pmod{4}$ and $3^m-1\equiv 0\pmod{8}$. Thus, we obtain  $4\cdot (3^h-1)_2=8\leq (3^m-1)_2$, which is precisely the inequality given in \eqref{dd}.

If $h\not\equiv0\pmod{M}$ and $h$ is even, then $M=2^w\cdot(h)_2$ where $w\geq 1$. Thus, we obtain 
\begin{align*}
&(3^m-1)_2=(3^{(m)_2}-1)_2=(3^{2M}-1)_2=(3^{2^{w+1}(h)_2}-1)_2\\
=&(3^{(h)_2}-1)_2 \cdot (3^{(h)_2(2^{w+1}-1)}+\cdots+3^{(h)_2}+1)_2\\
\geq & 4\cdot(3^{(h)_2}-1)_2.
\end{align*}

On the other hand, we have
\begin{align*}
(3^h-1)_2=(3^{(h)_2(h)_{2'}}-1)_2=(3^{(h)_2}-1)_2(3^{(h)_2((h)_{2'}-1)}+\cdots+3^{(h)_2}+1)_2=(3^{(h)_2}-1)_2.
\end{align*}

Thus, we obtain  $4\cdot (3^h-1)_2=4\cdot (3^{(h)_2}-1)_2\leq (3^m-1)_2$, which is precisely the inequality given in \eqref{dd}.

{\bf Case 3}. $m\equiv 2, 4, 6\pmod{8}$. 

If $h$ is odd, then $3^h-1\equiv 2\pmod{4}$ and $3^m-1\equiv 0\pmod{8}$. Thus, we obtain $4\cdot (3^h-1)_2=8\leq (3^m-1)_2$, which is precisely the inequality given in \eqref{dd}. 

If $h$ is even, then $3^h-1\equiv 0\pmod{8}$ and $3^m-1\not\equiv 0\pmod{32}$. Thus \eqref{dd} cannot happen.
This completes the proof. 
\end{proof}

Recall from Theorem \ref{thm_p1m1} that if $e$ satisfies $e(3^h+1)\equiv \frac{3^m-1}{2}\pmod{3^m-1}$, then $\mc{C}_{(1,e)}$ is never optimal. The situation is much different when $e$ satisfies $e(3^h-1)\equiv \frac{3^m-1}{2}\pmod{3^m-1}$. Computational evidence suggests that in this case optimal codes $\mc{C}_{(1,e)}$ do exist for specific choices of $e$. For example, when $m\in\{2,6,10,14,18\}$ and $h=\frac{m}{2}$, the exponents $e=\frac{3^h+1}{2}+(3^h+1)t$ (for certain $t$) yield optimal codes. This motivates the following problem.

\begin{problem}
Let $m$ be a positive integer with $m\equiv 2\pmod{4}$ and let $h=\frac{m}{2}$. Let $e=\frac{3^h+1}{2}+(3^h+1)t$ where $0\leq t\leq 3^h-2$. For which values of $t$, the codes $\mc{C}_{(1,e)}$ have parameters $[3^m-1,3^m-1-2m,4]?$
\end{problem}

\section*{Acknowledgement} 
The authors are very grateful to the reviewers and the editor for their detailed comments and suggestions that improved the presentation and quality of this paper.
Jingjun Bao acknowledges the support of the National Natural Science Foundation of China Grant No. 12471313. 
Hanlin Zou acknowledges the support of the National Natural Science Foundation of China Grant No. 12461061.

\bibliographystyle{plain}

\end{document}